\documentclass[letterpaper, 10 pt, conference]{ieeeconf}  

\IEEEoverridecommandlockouts                              
\usepackage{amsmath} 
\usepackage{amssymb}  
\usepackage{xcolor}

\usepackage{graphicx}

\usepackage{subcaption}

\newtheorem{definition}{Definition}

\newtheorem{problem}{Problem}

\newtheorem{assumption}{Assumption}
\newtheorem{theorem}{Theorem}

\title{\LARGE \bf
Decentralized Signal Temporal Logic Control for Perturbed Interconnected Systems via Assume-Guarantee Contract Optimization
}

\author{Kasra Ghasemi, Sadra Sadraddini, and Calin Belta
\thanks{This work was partially supported by the NSF under grants IIS-2024606 and IIS-1723995.}
\thanks{K. Ghasemi and C. Belta are with the Division of System Engineering,
        Boston University, Boston, MA 02215, USA
        {\tt\small kasra0gh@bu.edu, cbelta@bu.edu}}%
\thanks{S. Sadraddini is with the Computer Science and Artificial Intelligence Laboratory, Massachusetts Institute of Technology, Cambridge, MA 02139, USA
        {\tt\small sadra@mit.edu}}%
}

\begin{document}

\maketitle
\thispagestyle{empty}
\pagestyle{empty}

\begin{abstract}
We develop a novel decentralized control method for a network of perturbed linear systems with dynamical couplings subject to Signal Temporal Logic (STL) specifications. 
We first transform the STL requirements into set containment problems and then we develop controllers to solve these problems. Our approach is based on treating the couplings between subsystems as disturbances, which are bounded sets that the subsystems negotiate in the form of parametric assume-guarantee contracts. The set containment requirements and parameterized contracts are added to the subsystems' constraints. We introduce a centralized optimization problem to derive the contracts, reachability tubes, and decentralized closed-loop control laws. We show that, when the STL formula is separable with respect to the subsystems, the centralized optimization problem can be solved in a distributed way, which scales to large systems. We present formal theoretical guarantees on robustness of STL satisfaction. The effectiveness of the proposed method is demonstrated via a power network case study.
\end{abstract}

\section{INTRODUCTION}

Multi agent systems benefit from decentralized control laws that  require only local information. Online computations lessen as each system implements its own control law. Furthermore, communication issues are also mitigated as agents do not need to constantly share information. Multi agent control drew a lot of interest in control theory, particularly in the last two decades \cite{olfati2007consensus,cortes2004coverage, egerstedt2001formation}. For multi agent systems with hard constraints and uncertainties, rigorous mathematical tools are required to reason about the closed-loop performance of the aggregate system. 

Formal methods provide mathematical guarantees for the behavior of control systems. Formal languages, such as temporal logics \cite{baier2008principles}, can be used to describe system specifications. With particular relevance to this work, Signal Temporal Logic (STL) \cite{maler2004monitoring} can describe a broad range of temporally bounded constraints. For example, the STL formula $\psi={\textbf F}_{[0,10]} (x > 10) \vee {\textbf G}_{[5,20]} (x < 0) $ reads  in plain English ``\emph{eventually} between time 0 and 10 the value of $x$ exceeds 10 {or} \emph{always} between 5 and 20 the value of $x$ remains below 0". 
The use of formal methods in multi-agent systems has also been investigated \cite{pant2018fly,liu2017distributed,lindemann2019control}. But, with only one exception \cite{lindemann2019control}, they only studied dynamically decoupled agents, and none of them took into account the presence of additive disturbances.
A related approach in formal methods is based on set-valued dynamics. Analyzing such systems enables characterizing all the possible responses in the presence of bounded uncertainties. Reachability analysis and correct-by-design control synthesis, which guarantee correctness without the need for system testing, received a lot of attention in recent years \cite{girard2005reachability,majumdar2017funnel, dutta2019reachability}. 

Formal methods come with a high computational cost, which makes it challenging to apply them to multi agent systems. That is especially true when we are considering systems with disturbances, and want to guarantee the satisfaction of temporal logic specification under all allowed disturbances.  Divide and conquer techniques are a natural way to break the problem into smaller pieces. They can be applied to interconnected systems, where the dynamics of the agents are coupled. Assume-guarantee contracts \cite{chatterjee2007assume} formalize the promises that systems make and provide over dynamical couplings. For instance, assume-guarantee contracts were used to describe vehicular flow between neighborhoods of a traffic network \cite{kim2015compositional}, aircraft power distributions, \cite{oh2019optimizing}, and dynamics of an aerial robot tethered to a ground one \cite{nilsson2016synthesis}.  

In this paper, we study the problem of decentralized control design for interconnected perturbed linear systems subject to STL constraints. Unlike approaches that assume that feasible assume-guarantee contracts are given a-priori \cite{nuzzo2015platform, chen2019compositional}, we parameterize the contracts and search for feasibility. Unlike the search methods in \cite{kim2015compositional,lin2020decentralized}, our parameterization, which is based on our prior work \cite{ghasemi2020compositional}, has a special convexity property that leads to a tractable solution. The approach in \cite{nilsson2016synthesis} also parameterized contracts and found them using convex optimization, but was limited to polytopic invariant sets. Here we include complex, non-convex STL constraints, and retain the parameterization from \cite{ghasemi2020compositional}. The main contributions of this paper are as follows:
\begin{enumerate}
    \item By fixing the ``logical behavior" through solving a mixed-integer program, we are able to convert the STL specifications into set containment problems. Then, a linear program is proposed to jointly optimize assume-guarantee contracts, set-valued trajectories, and decentralized closed loop control laws. This allows steering the aggregate system in a way that the global STL formulae is satisfied, while disturbances are rejected in a decentralized manner. The resulting bounds are computed using assume-guarantee contracts and are connected to the STL robustness score, a signed distance to satisfaction.
    \item  When the given STL formula is separable with respect to the subsystems, we provide a method to make the contribution above computationally more tractable for large networks by making it compositional. We use the convexity properties in \cite{ghasemi2020compositional} to optimize contracts, reachability sets, and controllers in a distributed way.  
\end{enumerate}

The rest of the paper is organized as follows. We first provide the notation and the necessary background in Section II. We state the problem in Section III. The solution is provided in Sections IV and V. Finally, an illustrative example is shown in Section \ref{example}. 
\section{Notations and Preliminaries}

\subsection{Notation}
$\mathbb{R}$, $\mathbb{R}_+$ and $\mathbb{N}$ stand for the sets of real, non-negative real, and non-negative integers, respectively;  $\mathbb{N}_h$ represents the set of non-negative numbers up to $h\in \mathbb{N}$. An $h$-dimensional box is defined as $\mathbb{B}_h:=\{b\in \mathbb{R}^h| ||b||_\infty \leq 1\}$. $\mathbb{S}_1 \oplus \mathbb{S}_2 := \{s_1 + s_2 | s_1 \in \mathbb{S}_1 , s_2 \in \mathbb{S}_2 \}$ is the Minkowski sum of two sets $\mathbb{S}_1$ and $\mathbb{S}_2$. 
The \textit{Directed Hausdorff distance} $d_{DH}(\mathbb{S}_1,\mathbb{S}_2)$ is a quantitative measure of how far $\mathbb{S}_2$ is from being a subset of $\mathbb{S}_1 $, and it can be computed as:
\begin{equation} \label{directed_huasdorff_distance_def}
    d_{DH}(\mathbb{S}_1,\mathbb{S}_2) :=  \sup_{s_2\in \mathbb{S}_2} \inf_{s_1\in \mathbb{S}_1} d(s_1,s_2) ,
\end{equation}
where $d: \mathbb{R}^n \times \mathbb{R}^n \rightarrow \mathbb{R}_+$ is a metric. For compact sets, $d_{DH}(\mathbb{S}_1,\mathbb{S}_2)=0$ if and only if $\mathbb{S}_2 \subseteq \mathbb{S}_1$. The Cartesian product of sets $\mathbb{S}_1$ and $\mathbb{S}_2$ is denoted by $\mathbb{S}_1 \times \mathbb{S}_2$ and the Cartesian product of $\mathbb{S}_1, \cdots, \mathbb{S}_N$ by $\prod_{i=1}^N \mathbb{S}_i$. $I_n$, $0_n$, and $[A_1,A_2]$ represent the $n \times n$ identity matrix, the $n$-dimensional zero vector, and the horizontal concatenation of matrices $A_1$, $A_2$ with the same number of rows, respectively.

\subsection{Zonotopes}
A zonotope is a symmetric shape set representation defined as $\mathcal{Z}(c,G) := \{ c + Gb | \forall b \in \mathbb{B}_q \}$,
where $c \in \mathbb{R}^n$ and $G\in \mathbb{R}^{n \times q}$ $(n,q \in \mathbb{N})$ denote the zonotope's center and generator, respectively.
The order of the zonotope is equal to $\dfrac{q}{n}$. Zonotopes are convenient for set calculations, such as Minkowski sums and linear transformations. Given two sets $\mathbb{S}_1 = \mathcal{Z}(c_1,G_1)$ and $\mathbb{S}_2 = \mathcal{Z}(c_2,G_2)$, a matrix $A \in \mathbb{R}^{m \times n}$, and a vector $b\in \mathbb{R}^n$, where $c_1,c_2 \in \mathbb{R}^{n}$ and $G_1\in \mathbb{R}^{n\times q_1}$ , $G_2\in \mathbb{R}^{n\times q_2}$, we have $\mathbb{S}_1 \oplus \mathbb{S}_2 = \mathcal{Z}(c_1 + c_2 , [G_1,G_2])$ and $A\mathbb{S}_1 + b = \mathcal{Z}(Ac_1+b , AG_1)$.
\subsection{Specifications}
Signal Temporal Logic (STL) was introduced in \cite{maler2004monitoring} to specify Boolean and temporal properties of real-valued, time signals. A discrete-time signal is a function $s : \mathbb{N} \rightarrow \mathbb{R}^{q}$. We use $(s,[t_1,t_2])$ to denote the sequence $s(t_1) , ... , s(t_2)$ and $(s,t)$ for $(s,[t,\infty])$. An STL formula is defined with the following recursive grammar:
\begin{equation}
    \varphi ::= \pi | \neg \varphi | \varphi \wedge \psi | \varphi | \varphi \vee \psi | \textbf{F}_{[t_1,t_2]} \varphi | \textbf{G}_{[t_1,t_2]} \varphi | \varphi \textbf{U}_{[t_1,t_2]} \psi
\end{equation}
where $\pi$ is a predicate. 
All predicates are assumed to be linear in the form $ p(s) \leq c$ or $ p(s) \geq c$, with $c$ being a scalar and $p: \mathbb{R}^{q} \rightarrow \mathbb{R}$ being a linear function. Symbols $\neg$ , $\wedge$ , and $\vee$ denote Boolean negation, conjunction, and disjunction, respectively; $\textbf{F}_{[t_1,t_2]}$, $\textbf{G}_{[t_1,t_2]}$, and $\textbf{U}_{[t_1,t_2]}$ are temporal operators for ``eventually",``always", and ``until", respectively.
Also, $(s,t) \models \varphi$ denotes that signal $s$ satisfies formula $\varphi$ at time $t$, and 
$(s,t) \nvDash \varphi$ if this is not the case. 
\begin{definition}
The satisfaction of a formula by a signal $s$ at time $t$ is defined as follows:
\begin{itemize}
    \item $(s,t) \models (p(s) \geq c) \Leftrightarrow p(s(t)) \geq c$ , 
    \item $(s,t) \models (p(s) \leq c) \Leftrightarrow p(s(t)) \leq c$ , 
    \item $(s,t) \models \neg \varphi  \Leftrightarrow (s,t) \nvDash \varphi$ ,
    \item $(s,t) \models \varphi_1 \wedge \varphi_2 \Leftrightarrow (s,t) \models \varphi_1 \wedge (s,t) \models \varphi_2 $ ,
    \item $(s,t) \models \varphi_1 \vee \varphi_2 \Leftrightarrow (s,t) \models \varphi_1 \vee (s,t) \models \varphi_2 $ ,
    \item $(s,t) \models \textbf{G}_{[t_1,t_2]} \varphi  \Leftrightarrow \forall t' \in [t_1 , t_2] , (s,t) \models \varphi $ ,
    \item $(s,t) \models \textbf{F}_{[t_1,t_2]} \varphi  \Leftrightarrow \exists t' \in [t_1 , t_2] , (s,t') \models \varphi $ ,
    \item $(s,t) \models \varphi_1 \textbf{U}_{[t_1,t_2]} \varphi_2 \Leftrightarrow \exists t' \in [t_1 , t_2] , (s,t') \models \varphi_2 \wedge \forall t^{\prime\prime} \in [t_1 , t^{\prime\prime}](s,t^{\prime\prime}) \models \varphi_1 $.
\end{itemize}
\end{definition}
For simplicity, $(s,0)\models \varphi$ is denoted by $s\models \varphi$. The {\em horizon} of a formula is the shortest amount of time required to determine whether a formula $\varphi$ is satisfied, and it is denoted by $hrz(\varphi)$ \cite{Belta2018}. The {\em robustness} \cite{donze2010robust} of an STL formula with respect to a signal determines how strongly the signal satisfies / violates the formula. Robustness is a real function that produces a {\em score}, where larger scores mean stronger satisfaction. The robustness of the formula $\varphi$ with respect to the signal $s$ at time $t$ is denoted by $\rho(s,\varphi,t)$ and can be computed recursively.
The procedure begins with the predicates, where each predicate's robustness is defined as $ \rho(s , (p(s) \geq c) , t) =   \rho = p(s)-c $.
Without loss of generality, we only consider negation free formulas in this paper. This is not restrictive, as any STL formula can be made negation-free. It is also worth noting that, while predicates with inequalities are used in the semantics definition, strict inequalities and equalities can be formed using the Boolean operators.



\section{Problem Definition and Approach}
Consider the following network of coupled time-variant linear subsystems:
\begin{multline} \label{network_sys}
    x_i(t+1) =  A_{ii}(t)x_i(t) + B_{ii}(t)u_i(t) + \sum_{ j \neq i}A_{ij}(t)x_j(t) \\+ \sum_{ j \neq i}B_{ij}(t)u_j(t) + w_i(t),\;i\in\mathcal{I}, 
\end{multline}
where $\mathcal{I}$ is an index set for the subsystems; $A_{ii}(t) \in \mathbb{R}^{n_i \times n_i}$, $A_{ij}(t) \in \mathbb{R}^{n_i \times n_j}$, $B_{ii}(t) \in \mathbb{R}^{n_i \times m_i}$, and $B_{ij}(t) \in \mathbb{R}^{n_i \times m_j}$ are given, time-variant matrices for subsystem $i$. Let $\eta=|\mathcal{I}|$ denote the number of subsystems in the network. The state, control input, and disturbance for subsystem $i$ at time step $t$ are represented by $x_i(t) \in \mathbb{R}^{n_i}$, $u_i(t) \in \mathbb{R}^{m_i} $, and $w_i(t) \in \mathbb{R}^{n_i} $, which are bounded by given polytopic sets $x_i(t)\in X_i(t) \subseteq \mathbb{R}^{n_i}$, $u_i(t)\in U_i(t)\subseteq \mathbb{R}^{m_i}$, and $w_i(t) \in W_i(t) \subset \mathbb{R}^{n_i}$, respectively.
A {\em decentralized controller} $\mu_i(. , t) : X_i(t) \rightarrow U_i(t)$ is a function that maps the current state of subsystem $i$ into a control input in the control space of the same subsystem. System (\ref{network_sys}) with no disturbances is called a {\em nominal system}. 
\begin{definition}[Decentralized Finite-Time Viable Sets] \label{def:viable} Given $h \in \mathbb{N}$, the sequences of sets $\Omega_i(0),\Omega_i(1), ... , \Omega_i(h)$, $i \in \mathcal{I}$ for the interconnected system in \eqref{network_sys} are called decentralized \emph{viable} sets, if for all $t \in \mathbb{N}_h,\forall i \in \mathcal{I}$, $\Omega_i(t) \subseteq X_i(t)$ and there exists a set of policies $\mu_i(.,t)$ such that $\Theta_i(t)   \subseteq U_i(t)$ and $\forall t \in \mathbb{N}_{h-1}, \forall x_i(t)\in \Omega_i(t), \forall w_i(t)\in W_i(t) \Rightarrow x_i(t+1)\in \Omega_i(t+1)$, where $
\Theta_i(t):=\mu(\Omega_i(t),t)$ is called {\em action set}.
\end{definition}

A signal $s: \mathbb{N} \rightarrow X \times U \subset \mathbb{R}^{n+m}$ is a trajectory where $s(t)$ represents a vector stacking the state and control of the aggregated system at time step $t$, which is represented by $s(t) = (x(t) , u(t))$, where $x(t) = [x^T_1(t) , \cdots , x^T_{\eta}(t)]^T \in \mathbb{R}^n$ and $u(t) = [u^T_1(t) , \cdots , u^T_{\eta}(t)]^T \in \mathbb{R}^m$ and $n = \sum_{i\in \mathcal{I}}n_i$ and $m = \sum_{i\in \mathcal{I}}m_i$.


In this paper, we consider the following problem:

\vspace{2mm}

\begin{problem} \label{problem_1}
Given a network of perturbed linear systems in the form \eqref{network_sys}, the initial states $x_i^{initial}(0) \in X_i(0),\forall i \in \mathcal{I}$, a bounded STL formula $\varphi$ with linear predicates in the states and / or controls, and a quadratic cost $J : \mathbb{S} \rightarrow \mathbb{R}_+$, find the optimal decentralized controllers $\mu_i(x_i(t) , t), \forall i\in \mathcal{I}$ and their corresponding sequence of viable sets $\Omega_i(t)$ such that $J$ is minimized, $x_i(t)\in X_i(t)$, $u_i(t) \in U_i(t)$, and $s \models \varphi$. If such a signal does not exist, find $\Omega_i(t)$ corresponding to the maximum possible value of the robustness, i.e, find a signal with the least amount of violation.
\end{problem}

\vspace{2mm}

To solve Problem \ref{problem_1},  a two-step optimization-based approach is proposed. We begin by solving a mixed-integer program for the aggregated nominal system, which is constrained by the STL formula $\varphi$ \cite{Belta2018}. It allows us to determine active predicates at each time and convert the STL formula satisfaction into a set containment problems, which is shown to be a convex programming problem \cite{sadraddini2019linear}. In the second step, we take into account the additive disturbance, along with the set containment constraints, and we find a set of decentralized closed-loop controllers and viable sets. The technical details are explained in the next sections.

\section{Converting STL formulas into Set Containment Problems}

In this section, the method from \cite{Belta2018} is used to encode an STL formula into a mixed-integer linear program. Then, the set of predicates whose satisfaction corresponds to the maximum robustness for the nominal system is identified and transformed to a set containment problem.
\subsection{Encoding the STL Formulas}



Following the predicate-based encoding from \cite{Belta2018}, a binary variable $z_t^\pi \in \{0,1\}$ is dedicated to each predicate $\pi = (y \geq 0)$, which must be assigned to $1$ if the predicate is true, and to 0 otherwise. The relation between $z_t^\pi$, the robustness $\rho$, and $y_t$ is encoded as
\begin{equation}\label{big_M_constraint}
y_t + M( 1 - z_t^\pi ) \geq \rho \quad,\quad y_t - M z_t^\pi < \rho ,
\end{equation}
where $M$ is a sufficiently large number such that for all time steps, $M \geq \text{max\color{red} }y_i , i \in \mathbb{N}_{n_y}$. The equations in \eqref{big_M_constraint} enforce the binary variable $z_t^\pi$ to be equal to $1$ when $y_t \geq \rho$ and equal to $0$ when $y_t < \rho$. 
Disjunctions and conjunctions are captured by the following constraints:
\begin{equation} \label{z_combination_constraints}
    z = \bigwedge_{i=0}^{n_z}z_i \Rightarrow z \leq z_i , i \in \mathbb{N}_{n_z} , z = \bigvee_{i=0}^{n_z}z_i \Rightarrow z \leq \sum_{i=0}^{n_z}z_i,
\end{equation}
where $n_z \in \mathbb{N}$ and $z\in [0,1]$ is declared as a continuous variable. However, as the above equation shows, it can only take binary values. In \cite{karaman2008optimal}, \cite{raman2014model} upper-bounding constraints are added to create a necessary and sufficient condition:
\begin{subequations}\label{milp_encoding_stl}
\begin{equation}
    z = \bigwedge_{i=0}^{n_z}z_i \Leftrightarrow z \geq \sum_{i=0}^{n_z} z_i - n_z +1 , z \leq z_i , i \in \mathbb{N}_{n_z} 
\end{equation} 
\begin{equation}
    z = \bigvee_{i=0}^{n_z}z_i \Leftrightarrow z \geq z_i, i \in \mathbb{N}_{n_z}, z \leq \sum_{i=0}^{n_z}z_i.
\end{equation}
\end{subequations}
The upper-bound constraints are necessary when the specification does not include negation. $z_t^\varphi \in [0,1]$ is the variable that indicates whether $(s,t) \models \varphi$. A recursive translation of an STL formula is as follows:
\begin{multline} \label{phi_combination_constraints}
    \varphi = \bigwedge_{i=1}^{n_\varphi} \varphi_i \Rightarrow z_t^\varphi = \bigwedge_{i=1}^{n_\varphi}z_k^{\varphi_i}   ;  \varphi = \bigvee_{i=1}^{n_\varphi} \varphi_i \Rightarrow z_t^\varphi =  \bigvee_{i=1}^{n_\varphi} z_t^{\varphi_i} ; \\
    \varphi = G_I \psi \Rightarrow z_t^{\varphi} = \bigwedge_{t'\in I} z_{t'}^\psi ; \varphi = F_I \psi \Rightarrow z_t^\varphi = \bigvee_{t' \in I} z_{t'}^\psi ; \\
    \varphi = \psi_1 U_I \psi_2 \Rightarrow z_t^\varphi = \bigvee_{t' \in I}(z_{t'}^{\psi_2} \wedge \bigwedge_{t'' \in [t,t']} z_{t''}^{\psi_1}),
\end{multline}
where $n_\varphi \in \mathbb{N}$. Given a formula $\varphi$, the set of constraints recursively constructed by equations \eqref{big_M_constraint}, \eqref{milp_encoding_stl}, and \eqref{phi_combination_constraints} is denoted by $\mathcal{C}_\varphi$.
\begin{theorem} [Adapted from \cite{Belta2018}]
The following properties hold for the above mixed-integer linear program encoding:(i) $(s,t) \models \varphi$, if adding $z_t^\varphi=1$ and $\rho \geq 0$ to the constraints makes $\mathcal{C}_\varphi$ feasible, (ii) $(s,t) \nvDash \varphi$, if adding $z_t^{\varphi} = 1$ and $\rho \geq 0$ makes $\mathcal{C}_\varphi$ infeasible, (iii) the largest $\rho$ such that $z_t^\varphi = 1$ and $\mathcal{C}_\varphi$ is feasible is equal to the robustness.
\end{theorem}
It is shown in \cite{Belta2018} that when the STL formulas are negation free, $\rho$ equals robustness. As a result, it can be used as an objective function to maximize robustness.
\subsection{Set Containment for STL Formula Satisfaction} \label{G_phi}
The objective of this subsection is to get the set of $z_t^{\pi}$s equal to $1$, which are called active predicates, for the maximum robustness while considering the nominal system. If the disturbance bound is small enough, it can be assumed that the perturbed and nominal systems have the same set of active predicates, and a closed-loop controller can be found to ensure that the system's reachability set still satisfies those predicates.
We can do the synthesis for the aggregate nominal system \cite{Belta2018} rewritten as $x(t+1) = A(t) x(t) + B u(t)$ from \eqref{network_sys}, by using the STL satisfaction constraints introduced before:
\begin{equation}\label{centralized_synthesis_nominal}
\begin{aligned} 
&\max_{ x(t) ,u(t), z_t^\pi , \rho   } \quad  -J(s[0,hrz(\varphi)]) +  \mathcal{M}( |\rho| - \rho) \\
\textrm{s.t.} \quad & x(t+1) =  A(t) x(t) + B u(t), t\in \mathbb{N}_{hrz(\varphi)-1}\\
    & x(0) = [x^{initial}_1 , ... , x^{initial}_{\eta}] ,  \\
    & \mathcal{C}_\varphi , z_0^\varphi = 1.
\end{aligned}
\end{equation}
As long as robustness is positive, the proposed objective function minimizes the user defined cost function $J(.)$, which can be a regular quadratic function in the form of $\sum_{t=0}^{hrz(\varphi)} x(t)^TQx(t) + \sum_{t=0}^{hrz(\varphi)} u(t)^TRu(t)$. Otherwise, it maximizes robustness due to the effect of the large scalar $\mathcal{M}$ and finds the nominal trajectory with the least violation.

Each active predicate is actually a set, $y_t \geq \rho ,\forall y_t$, which must hold for all possible signals at time $t$.
By definition, we have $s(t) \in \prod_i \Omega_i(t) \times \prod_i \Theta_i(t)$. Assuming the set $\prod_i \Omega_i(t) \times \prod_i \Theta_i(t)$ is represented by a zonotopic set $\mathcal{Z}(c,G)$ (notation $t$ is removed for readability), then any possible signal must satisfy $\mathrm{e} \geq \rho ,\quad  \forall \mathrm{e} \in \mathcal{Z}(p( c ) , p(G))$.
Also, by definition, the zonotope $\mathcal{Z}(c,G)$ has the following upper and lower bounds $c - \sum_i|g_i| \leq \mathcal{Z}(c,G) \leq c + \sum_i|g_i|$, where $g_i$ is the $i$th column of $G$. Using these bounds, the satisfaction constraint for an active predicate would be:
\begin{equation} \label{true_predicate}
    - p( c ) +  \sum_i |p(G)_i|  \leq -\rho
\end{equation}
where $p(G)_i$ is the $i$th element of $p(G)$. 
\begin{theorem}
The constraint in \eqref{true_predicate} can be written as a set of linear constraints as follows:
\begin{equation}\label{true_predicate_linear}
    - p( c ) + \sum_i p'_i \leq -\rho , p'_i \geq p(G)_i , p'_i \geq  - p(G)_i.
\end{equation}
\end{theorem}
\begin{proof}
It can be easily seen that if such $p'_i$s exist, the following relation holds:
\begin{equation}
        - p( c ) +  \sum_i |p(G)_i| \leq - p( c ) + \sum_i p'_i \leq -\rho ,    
\end{equation}
which also satisfies the original constraint \eqref{true_predicate}. Also, because \eqref{true_predicate_linear} is the relaxed form of the original problem, if such $p'_i$s do not exist, the original problem is also infeasible.
\end{proof}
Finally, the set of linear constraints that guarantees any possible trajectories in viable and action sets satisfies the STL formula $\varphi$ is denoted by $\mathcal{G}_\varphi$.

\section{Computation of Viable Sets under Additive Disturbance}
The original problem has been transformed into a decentralized control synthesis problem with zonotopic set containment constraints. The latter problem was considered in \cite{ghasemi2020compositional}, where a compositional approach using assume-guarantee contracts is proposed. In this section, we give a brief overview of \cite{ghasemi2020compositional} and incorporate the linear constraints $\mathcal{G}_\varphi$ into its formulation.

\subsection{Decentralized Synthesis}
First, the subsystems are decoupled from each other by considering the effects of other subsystems as disturbances, and by making some assumptions on the operational sets of each subsystem, as follows:
\begin{equation} \label{subsystems}
    x_i(t+1) = A_{ii}(t)x_i(t) + B_{ii}(t)u_i(t) + w_i^{aug}(t),
\end{equation}
where $w_i^{aug}(t)$ is the augmented disturbance set on subsystem $i$, which belongs to:
\begin{equation} \label{disturbance_AG}
    w_i^{aug}(t) \in \bigoplus_{j \ne i} A_{ij}(t) \mathcal{X}_j(t) \oplus \bigoplus_{j \ne i} B_{ij}(t) \mathcal{U}_j(t) \oplus W_i(t) ,
\end{equation}
where $\mathcal{X}_j(t)$ and $\mathcal{U}_j(t)$ are assumed operational sets for the state and the control input of subsystem $j\in \mathcal{I}$. It can be seen that the performance of each subsystem affects the assumptions of the other subsystems. This give and take contracts are called assume-guarantee contracts.
\begin{definition}[Assume-Guarantee Contracts]
An assume-guarantee contract for subsystem $i \in \mathcal{I}$ is a pair $\mathcal{C}_i = (\mathcal{A}_i , \mathcal{G}_i)$, where:
\begin{itemize}
    \item The assumption $\mathcal{A}_i$ is the assumption set over the disturbance $w_i^{aug}(t) \in \mathcal{W}_i(t)$,
    \item The guarantee $\mathcal{G}_i$ is the promise of subsystem $i$ over its state and control input $x_i(t) \in \mathcal{X}_i(t) , u_i(t) \in \mathcal{U}_i(t)$.
\end{itemize}
As seen in \eqref{disturbance_AG}, the following relation holds between the guarantee of other subsystems $\mathcal{X}_j, \mathcal{U}_j , j \ne i$ and the assumption of subsystem $i$, $\mathcal{A}_i$:
\begin{equation} \label{assumption_guarantee relation}
     \mathcal{W}_i(t)= \bigoplus_{j \ne i} A_{ij}(t) \mathcal{X}_j(t) \oplus \bigoplus_{j \ne i} B_{ij}(t) \mathcal{U}_j(t) \oplus W_i(t)
\end{equation}
\end{definition}
The above zonotopic set is represented by $ W_i(t) = \mathcal{Z}( d^{w}_i(t) , G^{w}_i(t) ) $, where $d^{w}_i(t) \in \mathbb{R}^{n_i}$ and $G^{w}_i(t) \in \mathbb{R}^{n_i \times l(t)}$. Next, we define a parametric assume-guarantee contract, which is similar to the regular contract except that the sets $\mathcal{X}_i(t)$, $\mathcal{U}_i(t)$ are replaced with the parametric sets below:
\begin{subequations} \label{parameterization}
\begin{equation}
    \mathcal{X}_i(t , \alpha_i^x(t)) := \mathcal{Z}( c_i^x(t), G_i^x \text{Diag}(\alpha_i^x(t)) ) ,
\end{equation}
\begin{equation}
    \mathcal{U}_i(t , \alpha_i^u(t)) := \mathcal{Z}( c_i^u(t), G_i^u \text{Diag}(\alpha_i^u(t)) ) ,
\end{equation}
\end{subequations}
where $G_i^x \in \mathbb{R}^{n_i \times f_i}$ and $G_i^u \in \mathbb{R}^{m_i \times g_i}$ ($f_i,g_i\in \mathbb{N}$) are given matrices defined by the user, and the vectors $c_i^x(t) \in \mathbb{R}^{n_i}$ , $\alpha_i^x(t) \in \mathbb{R}^{f_i}$, $c_i^u(t) \in \mathbb{R}^{m_i}$, and $\alpha_i^u(t) \in \mathbb{R}^{g_i}$ are parameters. Also, the parametric assumption set $\mathcal{W}_i(t , \alpha^{ext} )$ is derived by replacing the above parametric sets into equation \eqref{assumption_guarantee relation}, where $\alpha^{ext}$ denotes the set of all parameters.
To deal with the mismatch between the assumed and real operational disturbance sets, we introduce the notion of correctness:
\begin{definition}[Correctness]
A set of parametric contracts $\mathcal{C}_i$ is correct if 
\begin{equation} \label{correctnesss_critoria}
 \mathcal{W}_i(t) \subseteq \bigoplus_{j \ne i} A_{ij}(t) \Omega_j(t) \oplus B_{ij}(t) \Theta_j(t) \oplus W_i(t), \forall i,t.  
\end{equation}
\end{definition}
The preceding definition is required to resolve the circularity problem of assumption-guarantee contracts. It was shown in \cite{ghasemi2020compositional} that the following sufficient constraints imply \eqref{correctnesss_critoria}:
\begin{equation}
    \mathcal{X}_i(t , \alpha_i^x(t)) \subseteq \Omega_i(t) , \mathcal{U}_i(t , \alpha_i^u(t)) \subseteq \Theta_i(t). 
\end{equation}

The next step is to design a robust controller for each subsystem. The following decentralized controller structure is proposed for each subsystem:
\begin{equation}\label{controller}
    x_i(t) = \bar{x}^i_t + T^i_t \zeta ,  u_i(t) = \bar{u}^i_t + M^i_t \zeta , \zeta \in \mathbb{B}_k ,
\end{equation}
where $\bar{x}^i_t \in \mathbb{R}^{n_i}$, $\bar{u}^i_t \in \mathbb{R}^{m_i}$, $T^i_t \in \mathbb{R}^{n_i \times q_i(t)}$, and $M^i_t \in \mathbb{R}^{m_i \times q_i(t)}$ are unknowns that need to be tuned and $q(t) = k + \sum_{i=0}^{i=t}{l(i)}$, where $k \in \mathbb{N}$ is a hyper-parameter. 
Then, for subsystem $i$, it can be shown that the following linear constraints are sufficient for tuning the control parameters:
\begin{subequations} \label{viable_constraints}
\begin{equation} \label{viable_constraint_set}
    [ A_{ii}(t)T^i_t + B_{ii}(t)M^i_t , G^{aug}_i(t) ] = [T^i_{t+1}], t\in \mathbb{N}_{h-1}
\end{equation}
\begin{equation} \label{viable_constraint_center}
    A_{ii}(t) \bar{x}^i_t + B_{ii}(t) \bar{u}^i_t + d^{aug}_i(t) = \bar{x}^i_{t+1}, t\in \mathbb{N}_{h-1}.
\end{equation}
\end{subequations}
If such parameters exist, $\Omega_i(t) = \mathcal{Z}( \bar{x}^i_t , T^i_t ) $ is the viable set, $\Theta_i(t) = \mathcal{Z}( \bar{u}^i_t , M^i_t ) $ is the action set, and \eqref{controller} is the controller. Intuitively, constraints \eqref{viable_constraint_set} and \eqref{viable_constraint_center} are set containment constraints; \eqref{viable_constraint_center} adjusts the centers of the viable sets and \eqref{viable_constraint_set} takes care of the set expansion at each step, such that all the possible trajectories are contained within the tube $\Omega_i(0) , \Omega_i(1) , \cdots , \Omega_i(h)$, where $h$ is the horizon, which is set to $hrz(\varphi)$ in our problem.
Additionally, the following constraints are proposed to impose hard constraints:
\begin{equation}\label{hard_constraint_impose}
     \mathcal{Z}( \bar{x}^i_t , T^i_t ) \subseteq X_i(t), \mathcal{Z}( \bar{u}^i_t , M^i_t ) \subseteq U_i(t), t\in \mathbb{N}_{h}.
\end{equation}

It was demonstrated in \cite{sadraddini2019linear} that zonotope and polytope containment problems can be encoded into linear constraints.
Thus, all of the suggested constraints \eqref{correctnesss_critoria}, \eqref{viable_constraints}, \eqref{hard_constraint_impose}, and $\mathcal{G}_\varphi$ for all subsystems and time steps may be merged to build a centralized linear program to solve Problem \ref{problem_1}.
The objective function is ad-hoc, but we recommend the mean square error between the center line of viable/action sets and the nominal trajectory/controllers generated by \eqref{centralized_synthesis_nominal}.


\subsection{Compositional Computation of Decentralized Viable Sets}

Despite the fact that the centralized solution presented at the end of the preceding subsection is a linear program, it still suffers from curse of dimensionality in high dimensions. Nevertheless, it is demonstrated in \cite{ghasemi2020compositional} that the suggested parameterization \eqref{parameterization} allows for compositional computation of viable sets in a time-efficient manner by transforming a single, large linear program into a group of smaller linear programs. We show that if the STL formula in Problem \ref{problem_1} is separable by subsystems, we can also use the parameterization to solve the same centralized approach in the previous section in a compositional manner. Additionally, convergence is ensured due to the convexity of the problem set.
\begin{assumption}
The STL formula in Problem \ref{problem_1} is separable by the subsystems, meaning it can take the form $\varphi = \varphi_1 \wedge ... \wedge \varphi_\eta$, where $i$ is the subsystem's index.
\end{assumption}
In \cite{ghasemi2020compositional}, we proposed a parametric potential function that quantifies how far a set of contracts is from correctness. This comes in contrast to the previously introduced correctness property, which was either true or false. The larger the parametric potential function, the farther the set of contracts is from correctness, so the goal is to minimize the proposed potential function.
Here, the parametric potential function is modified by including the containment constraints coming from the STL formulas, as well as adding the sum of the directed Hausdorff distances between the hard constraints and the viable/action sets in \eqref{hard_constraint_impose} into the potential function.
\begin{definition}[Parametric potential function]
The parametric potential function $\mathcal{V}(\alpha^{ext})$ is defined as $\mathcal{V}(\alpha^{ext}) = \sum_{i\in \mathcal{I}} \mathcal{V}_i(\alpha^{ext})$, where
\begin{multline}\label{V_i_dh}
    \mathcal{V}_i(\alpha^{ext}) := \\ \sum_{t\in \mathbb{N}_{hrz(\varphi)}} [ d_{DH}(\mathcal{X}_i(t) , \Omega_i(t) )  + d_{DH}(\mathcal{U}_i(t) , \Theta_i(t) ) \\+ d_{DH}(X_i(t) , \Omega_i(t) ) + d_{DH}(U_i(t) , \Theta_i(t) ) ].
\end{multline}
\end{definition}
Using the technique explained in Subsection \ref{G_phi}, the satisfaction of the STL formula $\varphi_i$ for subsystem $i$ can be encoded as a set of linear constraints denoted by $\mathcal{G}_{\varphi_i}$. Each component of the parametric potential function $\mathcal{V}_i(\alpha^{ext})$ can be computed using these constraints and \eqref{viable_constraints} by solving the following linear program:
\begin{subequations} \label{V_i}
\begin{align}
    & \mathcal{V}_i(\alpha) = \min_{\underset{, d_t^x , d_t^u , \bar{d}_t^x , \bar{d}_t^u}{\mathrm{x}^i,T^i,\mathrm{u}^i,M^i } }
    \begin{aligned}[t]
       & \sum_{t \in \mathbb{N}_{hrz(\varphi)}} [d^x_t + \bar{d}^x_t] +  \sum_{t \in \mathbb{N}_{hrz(\varphi)-1}} [d^u_t + \bar{d}^u_t]
    \end{aligned} \notag \\
    &\text{subject to} \notag \\
    & [A_{ii}(t)T^i_t+ B_{ii}(t) M^i_t , G_i^{w}(t) ] =  [T^i_{t+1}],  \forall t \in \mathbb{N}_{hrz(\varphi)-1} \label{eq:a} \\
    & A_{ii}(t)\bar{\mathrm{x}}^i_t + B_{ii}(t) \bar{\mathrm{u}}^i_t + d_i^{w}(t) = \bar{\mathrm{x}}^i_{t+1},   \forall t \in \mathbb{N}_{hrz(\varphi)-1} \label{eq:b}\\
    & \mathcal{Z}(\bar{\mathrm{x}}^i_t,T^i_t) \subseteq \ \mathcal{X}_i(t,\alpha^x_i(t)) \oplus \mathcal{Z}(0,d^x_t I_{n_i}) ,  \forall t \in \mathbb{N}_{hrz(\varphi)} \label{eq:c}\\
    & \mathcal{Z}(\bar{\mathrm{u}}^i_t,M^i_t) \subseteq \ \mathcal{U}_i(t,\alpha^u_i(t)) \oplus \mathcal{Z}(0,d^u_t I_{m_i}) ,  \forall t \in \mathbb{N}_{hrz(\varphi)-1} \label{eq:d}\\
    & \mathcal{Z}(\bar{\mathrm{x}}^i_t,T^i_t) \subseteq X_i(t) \oplus \mathcal{Z}(0,\bar{d}^x_t I_{n_i}) ,  \forall t \in \mathbb{N}_{hrz(\varphi)} \label{eq:e}\\
    & \mathcal{Z}(\bar{\mathrm{u}}^i_t,M^i_t) \subseteq U_i(t) \oplus \mathcal{Z}(0,\bar{d}^u_t I_{m_i}) , \forall t \in \mathbb{N}_{hrz(\varphi)-1} \label{eq:f}\\
    & \mathcal{G}_{\varphi_i} , \bar{\mathrm{x}}^i_0 = x_i^{initial}(0)  \label{eq:g}\\
    & d_t^x , \bar{d}_t^x\geq 0 , \hspace{2 mm} \forall t \in \mathbb{N}_{hrz(\varphi)}, \label{eq:hausdorffdistance_x}\\
    & d_t^u, \bar{d}_t^u \geq 0 , \hspace{2 mm} \forall t \in \mathbb{N}_{hrz(\varphi)-1}. \label{eq:hausdorffdistance_u}
\end{align}
\end{subequations}
Constraints \eqref{eq:a} and \eqref{eq:b} originate from \eqref{viable_constraints}. Also, the STL satisfaction constraints and the initial state constraint are added in \eqref{eq:g}.
The remaining constraints along with the objective function compute an over-approximation for the summation of the directed Hausdorff distance between sets $\Omega_i(t)$ and $\mathcal{X}_i(t) / X_i(t)$ and $\Theta_i(t)$ and $\mathcal{U}_i(t)$/$U_i(t)$ over all time steps. This approach of computing the Directed Hausdorff distance is inspired from \cite{sadraddini2019linear}.
\begin{theorem}[Convexity of the potential function]
The potential function proposed above is convex with respect to the parameters. The set of acceptable parameters (correct and valid) is also a convex set.
\end{theorem}
\begin{proof}
As seen in \eqref{V_i}, each component of the potential function is a linear program, which makes $\mathcal{V}_i(\alpha^{ext})$ a convex and piecewise affine function (a sum of convex functions is convex). Also, it is a well-known fact that the level set of a convex function is a convex set, thus, the set of acceptable parameters, which is equal to the zero level set of the potential function, is also a convex set. 
\end{proof}
The idea is to minimize the potential function using gradient descent and iteratively update the parameters:
\begin{equation} \label{GD}
    \alpha^{ext} \leftarrow \alpha^{ext} -  \sum_{i\in \mathcal{I}} \nabla_{\alpha^{ext}} \mathcal{V}_i(\alpha^{ext}),
\end{equation}
Convergence to the global minimum is guaranteed because the proposed potential function is convex. Each subsystem can find the direction that is best for it ($\nabla_{\alpha^{ext}}\mathcal{V}_i(\alpha^{ext})$), using its own local information and the common knowledge parameters, breaking the problem down into many smaller linear programs. If the minimum of the potential function is zero (by definition, the potential function is always larger than zero), it indicates that both the set of derived parametric contracts are correct and the viable and actions sets are within hard constraints. Thus, the desired control policies and viable sets are determined. Also, the nominal trajectories and controllers derived from \eqref{centralized_synthesis_nominal} can be used as initial values for the center parameters in our parameterized sets to give the gradient descent a warm start.
\section{Case Study}
\label{example}
We apply the method developed in this paper to address the load-frequency problem in power networks \cite{lfc}. A network is made up of several areas, each with its own power generator and demands, and some of them can be connected to each other to interchange power as needed, depending on the network architecture. Each area's state is represented by a 2-dimensional vector $[\delta_i(t), f_i(t)]^T$, where $\delta_i(t) \in \mathbb{R}$ is the deviation of the phase angle and $f_i(t) \in \mathbb{R}$ is the deviation of the frequency at time $t$ for area $i\in \mathbb{N}$. Also, $u_i(t)\in \mathbb{R}$ is the control input, which is the amount of change from its nominal value in the power generated by the generator at the area $i$ and time $t$. The dynamics for each area is given by:
\begin{multline}
    \dot{\delta}_i(t) = 2 \pi f_i(t) \hspace{2mm} , \hspace{2mm} \dot{f}_i(t) = - \dfrac{f_i(t)}{T_{p_i}} + \dfrac{K_{p_i} u_i(t) }{T_{p_i}} - \\   \dfrac{K_{p_i}}{2 \pi T_{p_i}} ( \sum_{j\in \mathcal{N}_i} K_{s_{ij}} [ \delta_i(t) - \delta_j(t) ] ) - \dfrac{K_{p_i} \omega_i(t)}{T_{p_i}},
\end{multline}
where $K_{p_i}, K_{s_{ij}} , T_{p_i}$ are the system gain, synchronizing coefficient between area $i$ and $j$, and system model time constant. In this case study, they are set to 110, 0.5, and 25, respectively, for all areas. Also, $\omega_i(t)$ is the load disturbance for area $i$ at time $t$, which is bounded by $|\omega_i(t)| \leq 0.001$. In addition, $\mathcal{N}_i$ denotes the neighbours of area $i$. Here, we consider the ring network architecture consisted of $20$ areas. Also, the control input is bounded by $|u_i(t)| \leq 0.1$. We use the Euler method to discretize the dynamics for every $0.1$ unit of time. For all areas, the initial state is  $[0.1 , 0.1 ]^T$ and the STL specification is $\varphi_i = \textbf{F}_{[0,6]}\textbf{G}_{[0,2]} \psi_1  \wedge \textbf{F}_{[0,8]} \psi_2$,
where $\psi_1 = [ \delta_i \leq 0.26] \wedge [\delta_i \geq 0.14] \wedge [f_i \leq -0.04 ] \wedge [f_i \geq -0.16 ]$ and $\psi_2 = [ \delta_i \leq 0.01] \wedge [\delta_i \geq -0.01] \wedge [f_i \leq 0.01 ] \wedge [f_i \geq -0.01 ]$. The goal is to synthesize decentralized controllers for each area subject to the specifications. 
We set the horizon to nine and synthesize the controllers using our approach. The baseline parametric sets are selected to be the viable and action sets generated from \eqref{viable_constraints} while couplings to other areas are ignored. The initial value of all parameters in the distributed algorithm is one. We used Gurobi on a MacBook Pro with 2.6 GHz 6-Core Intel Core i7 and 16 GB memory to run the algorithm. The results are shown in Fig.\ref{viable_fig}, and Fig.\ref{control_fig}. It can be seen that any possible trajectory that passes through the viable sets satisfies the STL specification and at the same time all the implemented controllers satisfy the hard constraint on the control input.

To demonstrate the approach's scalability, we experimented with various number of areas in the ring network and reported the running time in Fig \ref{fig_time}. The stated time period includes only the time spent on second step, but not on solving the MILP. That is because both distributed and centralized approaches share the first step. Additionally, to ensure that the solution exists for high-dimensional state spaces, we consider a large bound for the controller(i.e. $|u_i(t)| \leq 10$). As predicted, the distributed approach has a slower growth rate, making it more appropriate for the state spaces larger than $40$. Moreover, one of the primary benefits of the distributed technique is that it may be calculated in parallel. While we handled everything sequentially here, if multiprocessing is employed, the stated time could be reduced more depending on the number of cores used.
\begin{figure*}[h!] 
  \centering
  \begin{subfigure}[b]{0.34\linewidth}
    \includegraphics[width=\linewidth]{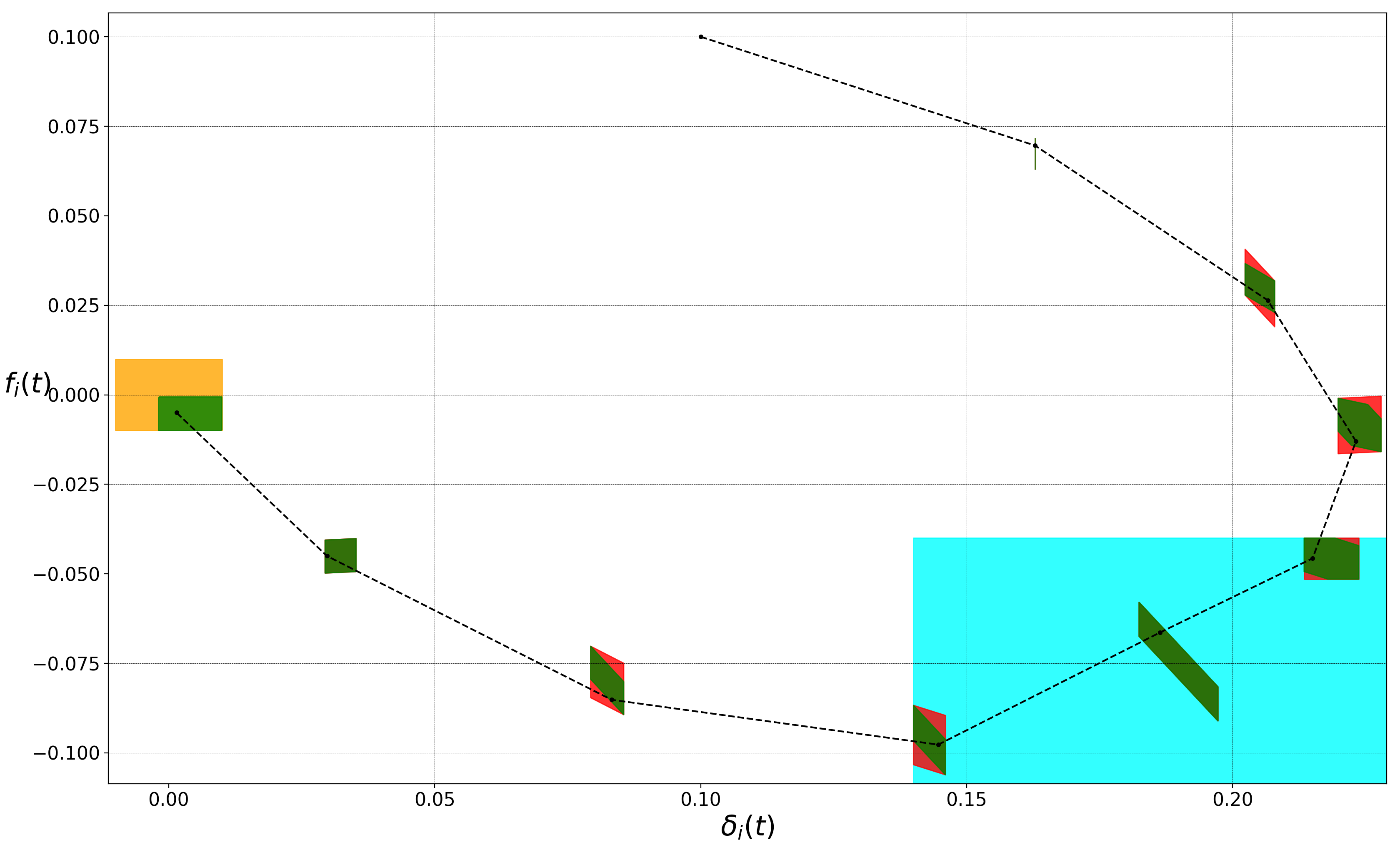}
    \caption{}
    \label{viable_fig}
  \end{subfigure}
  \begin{subfigure}[b]{0.34\linewidth}
    \includegraphics[width=\linewidth]{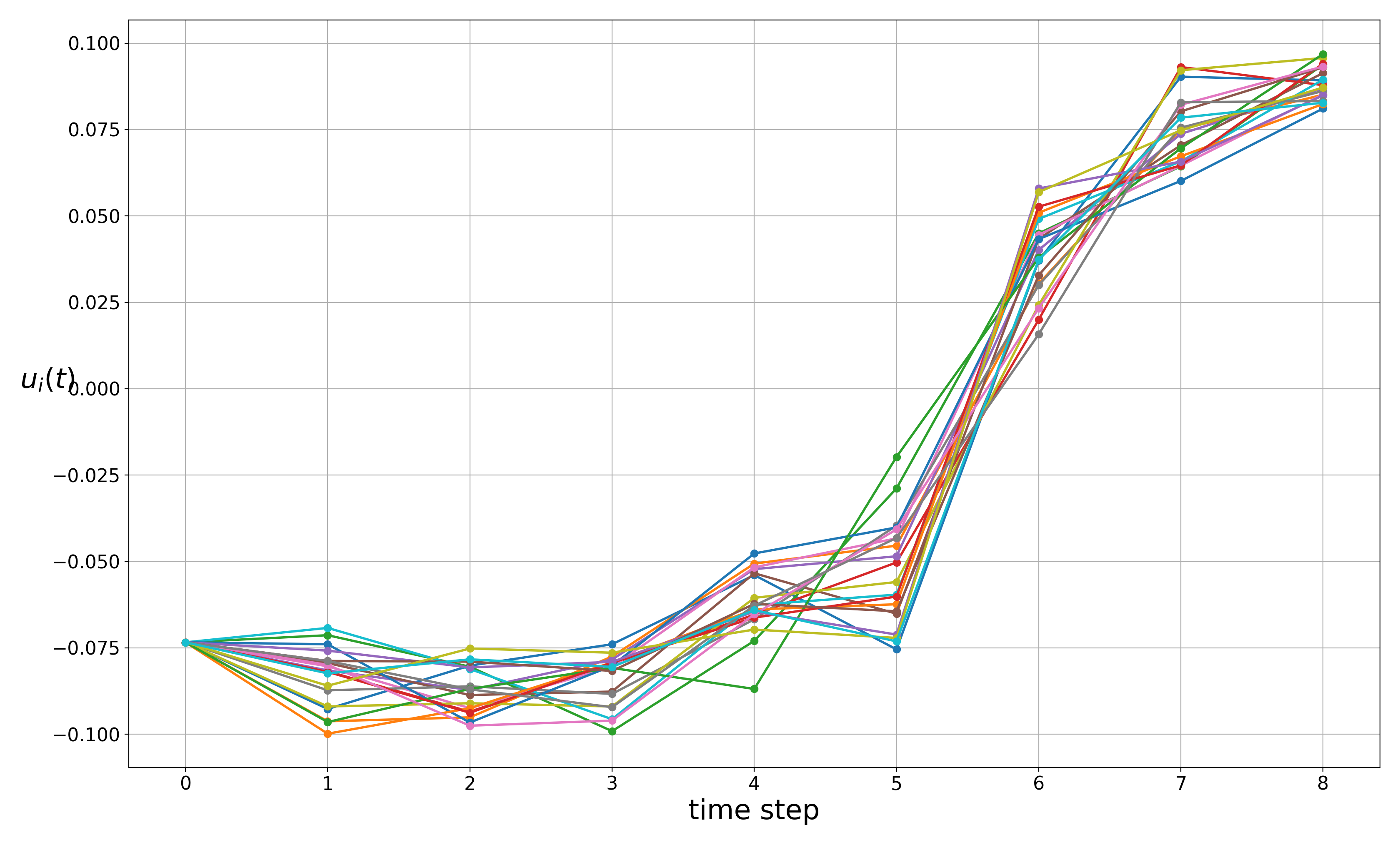}
    \caption{}
    \label{control_fig}
  \end{subfigure}
  \begin{subfigure}[b]{0.29\linewidth} 
    \includegraphics[width=\linewidth]{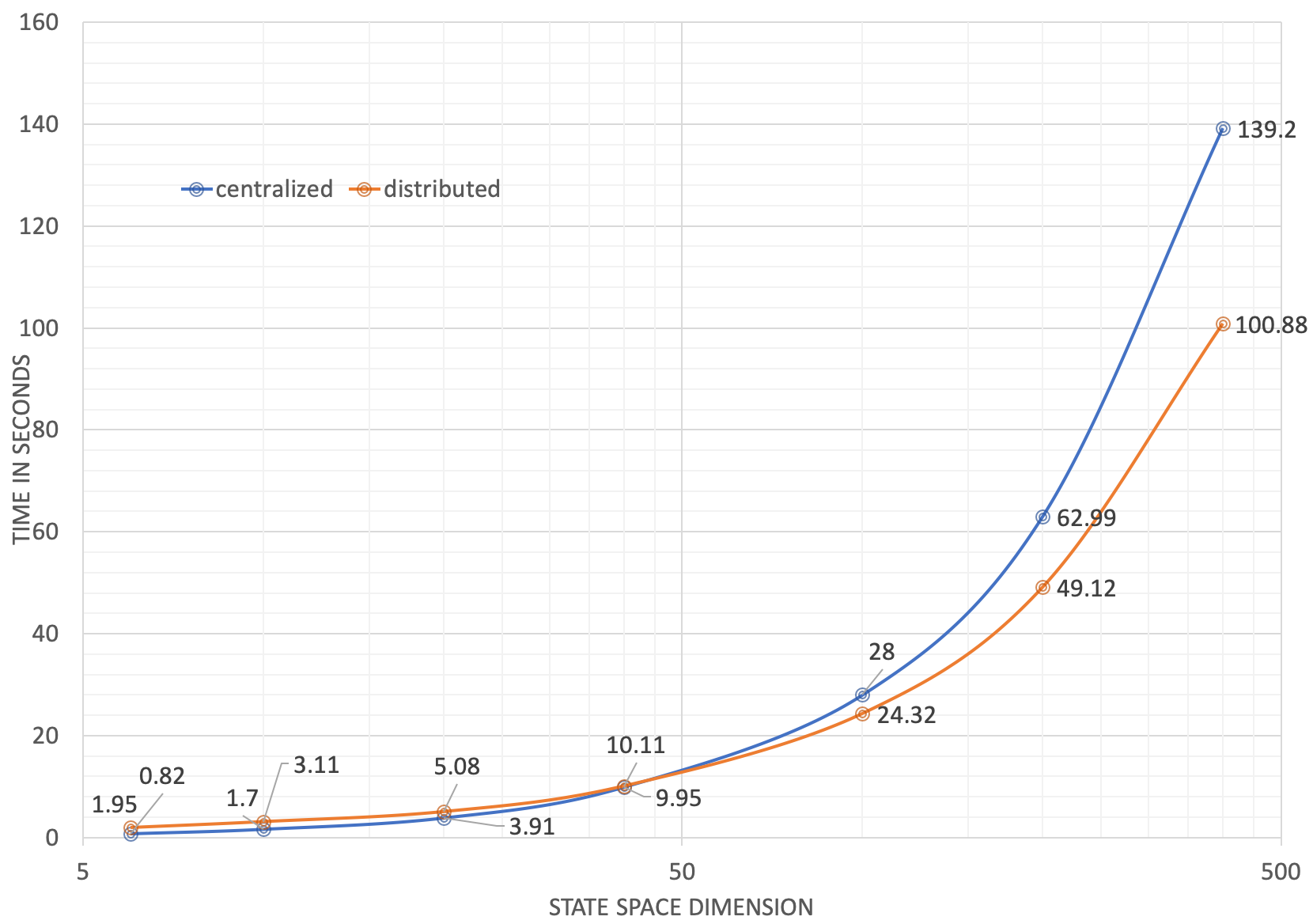}
    \caption{}
    \label{fig_time}
  \end{subfigure}
  \caption{(a) The green sets illustrate the viable sets for one of the areas in the case study. The blue and orange sets are the set of states satisfying $\psi_1$ and $\psi_2$, respectively. The red sets show the parameterized sets defined on the state space at different time steps for this specific area (some of them are tightly close to the viable sets and are not visible). The black line represents the trajectory traveled by this area. (b) Controllers for each area as time series. (c) Reported time in seconds for the distributed and centralized approach for different state space dimensions.}
  \label{fig}
\end{figure*}
\section{CONCLUSIONS}
Control synthesis subject to both a STL formula and a bounded disturbance is a computationally challenging problem. To overcome this challenge, we propose a solution which consists of two steps: First, we convert satisfaction of the STL formula into a set containment problem. To handle it, we consider the nominal system and use a centralized MILP. We claim that for small enough disturbances, both systems would have the same set of active predicates, which are seen as bounds. Second, we synthesize controllers subject to these bounds. Since the second step needs a set-based calculation, it has a relatively higher computational cost and thus creates a bottleneck for large scale systems. We show that this step can be achieved in a compositional fashion when the STL formula is separable by subsystems.

In the future, we will investigate the possibility of replacing the MILP in the first step, which remains a barrier due to its computational cost. Sampling methods are a promising direction. Additionally, we are considering employing a distributed control architecture instead of the decentralized architecture proposed here. 




\bibliographystyle{plain}        
\bibliography{references}           
                                 
\end{document}